\documentclass[conference,10pt]{IEEEtran}
\usepackage{graphicx,color}
\usepackage{amsmath, amsthm, amsfonts, amssymb, amsbsy,nccmath,bbm}
\usepackage{algorithm}
\usepackage{algpseudocode}
\usepackage{enumerate}
\usepackage{lipsum}
\newtheorem{lemma}{Lemma}

\usepackage[sort,compress]{cite}
\usepackage{epsfig}
\usepackage{epstopdf}
\usepackage{mathtools}
\usepackage{dsfont}
\usepackage{epstopdf}

\usepackage[inline]{enumitem}   
\makeatletter
\newcommand{\inlineitem}[1][]{%
\ifnum\enit@type=\tw@
    {\descriptionlabel{#1}}
  \hspace{\labelsep}%
\else
  \ifnum\enit@type=\z@
       \refstepcounter{\@listctr}\fi
    \quad\@itemlabel\hspace{\labelsep}%
\fi}
\makeatother
\newtheorem{theorem}{Theorem}

\newcommand{\beq}{\begin{equation}}
\newcommand{\eeq}{\end{equation}}

\def\adots{\mathinner{\mskip0mu\raise0pt\vbox{\kern7pt\hbox{.}}\mskip3mu
          \raise4pt\hbox{.}\mskip3mu\raise8pt\hbox{.}\mskip0mu}}

\newcommand{\diag}{\mbox{diag}}

\usepackage{bm}

\newcommand{\bmtheta}{{\bm \theta}}

\newcommand{\bmy}{{\bm y}}
\newcommand{\bmv}{{\bm v}}

\newcommand{\bmP}{{\bm P}}

\newcommand{\bmA}{{\bm A}}

\newcommand{\bmC}{{\bf C}}

\newcommand{\bmI}{{\bf {I}}}

\newcommand{\bmzero}{{\bm 0}}

\newcommand{\bma}{{\bm a}}

\newcommand{\thetah}{\widehat{\theta}}

\newcommand{\bmeta}{\boldsymbol{\eta}}

\makeatletter
\newcommand\fs@spaceruled{\def\@fs@cfont{\bfseries}\let\@fs@capt\floatc@ruled
  \def\@fs@pre{\vspace{0.5\baselineskip}\hrule height.8pt depth0pt \kern2pt}%
  \def\@fs@post{\kern1pt\hrule\relax}%
  \def\@fs@mid{\kern2pt\hrule\kern2pt}%
  \let\@fs@iftopcapt\iftrue}
\makeatother

\newcommand{\bit}{\begin{itemize}}
\newcommand{\eit}{\end{itemize}}

\renewcommand{\bmA}{{\mathbf A}}

\newcommand{\rank}{\mbox{rank}}
\newcommand{\bmB}{{\mathbf B}}

\newcommand{\bmb}{{\mathbf b}}

\newcommand{\rhob}{{\overline{\rho}}}

\usepackage{lipsum}

\renewcommand{\bmzero}{{\boldsymbol 0}}
\newcommand{\bmone}{{\boldsymbol 1}}

\DeclareMathOperator{\E}{\mathbb{E}}
\newcommand{\var}{\mbox{var}}

\newcommand{\bh}{{\widehat{b}}}

\newcommand{\bmAb}{{\overline{\bmA}}}
\newcommand{\bmthetab}{{\overline{\bmtheta}}}
\newcommand{\bmd}{{\mathbf{d}}}
\newcommand{\bmO}{{\mathbf{O}}}
\newcommand{\proj}{{\mathrm{proj}}}
\newcommand{\bmyb}{{\overline{\bmy}}}

\usepackage{subfiles}

\usepackage[acronym]{glossaries}

\newacronym{kld}{KLD}{Kullback–Leibler divergence}
\newacronym{snr}{SNR}{signal-to-noise ratio}
\newacronym{ap}{AP}{access point}

\newtheorem{proposition}{Proposition}

\begin{document}
\linespread{0.82}

\title{Approximating Univariate Factored Distributions via Message-Passing Algorithms}
\author{%
  \IEEEauthorblockN{Zilu Zhao,  Dirk Slock}
  \IEEEauthorblockA{
			\small
			Communication Systems Department, EURECOM, France \\	
			zilu.zhao@eurecom.fr, dirk.slock@eurecom.fr
			\vspace{-3mm}		}
	}

\maketitle

\begin{abstract}

Gaussian Mixture Models (GMMs) commonly arise in communication systems, particularly in bilinear joint estimation and detection problems. Although the product of GMMs is still a GMM, as the number of factors increases, the number of components in the resulting product GMM grows exponentially.
To obtain a tractable approximation for a univariate factored probability density function (PDF), such as a product of GMMs, we investigate iterative message-passing algorithms. Based on Belief Propagation (BP), we propose a Variable Duplication and Gaussian Belief Propagation (VDBP)-based algorithm.
The key idea of VDBP is to construct a multivariate measurement model whose marginal posterior is equal to the given univariate factored PDF. We then apply Gaussian BP (GaBP) to transform the global inference problem into local ones. 
Expectation propagation (EP) is another branch of message passing algorithms. In addition to converting the global approximation problem into local ones, it features a projection operation that ensures the intermediate functions (messages) belong to a desired family. Due to this projection, EP can be used to approximate the factored PDF directly. However, even if every factor is integrable, the division operation in EP may still cause the algorithm to fail when the mean and variance of a non-integrable belief are required. Therefore, this paper proposes two methods that combine EP with our previously proposed techniques for handling non-integrable beliefs to approximate univariate factored distributions.

\end{abstract}

\section{Introduction}
The Gaussian Mixture Model (GMM) commonly appears in communication problems. In bilinear joint channel estimation and data detection problems, the channel likelihood can be characterized as a GMM \cite{10378663}.
If the data sequence length is greater than one, this GMM becomes the product of multiple simple GMM factors. As the data length increases, the number of simple GMM factors composing the product GMM grows linearly. As a result, the number of components in the product GMM grows exponentially with the number of factors.

Although the moments of a single simple GMM with a limited number of components are easy to obtain, the moments of the product GMM in the bilinear estimation and detection problems become intractable as the number of factors increases. In \cite{6898015, 9460784, 10942970}, the authors attempt to avoid GMMs in bilinear joint estimation and detection using various methods. Distributions in the form of a product of GMMs can be generalized as univariate factored distributions. This paper investigates iterative message-passing-based algorithms to approximate the statistics of univariate factored distributions.

\subsection{Prior Works}
Given a univariate factored distribution, message-passing algorithms can be used to obtain the marginal distributions, which is equivalent to finding a stationary point of a constrained Bethe Free Energy (BFE) minimization \cite{heskes2005approximate}.

\subsubsection{Belief Propagation}
Belief Propagation (BP) was introduced by Pearl as an approximate inference tool \cite{pearl2014probabilistic}. Within the BFE framework, BP aims to find a stationary point of the BFE under strict marginal consistency constraints. BP is defined on the basis of a factored probability density function (PDF). A corresponding factor graph can be constructed by representing all factors and variables as nodes, and connecting a factor node to a variable node whenever the variable appears as an argument of that factor. By iteratively applying the sum–product rules, BP approximates the marginal distributions via auxiliary functions, called messages \cite{8187302, murphy2013loopybeliefpropagationapproximate}.

From the constrained BFE optimization perspective, BP seeks a stationary point of the BFE under strict marginal consistency constraints. This requirement often results in intractable functional forms for the messages. Consequently, even though BP reduces the dimensionality of the inference problem, its computational complexity may still be prohibitive when the messages are intractable—for example, when the factored pdf is a product of GMMs. \cite{ihler2003efficient, 1211409}

Gaussian BP (GaBP) \cite{roche2025affine, 9460784} is a variant of BP that projects messages to Gaussian distributions based on the Central Limit Theorem (CLT).

\subsubsection{Expectation Propagation}

Expectation Propagation (EP) \cite{minka2013expectationpropagationapproximatebayesian} was proposed to handle intractable messages by introducing a projection step that approximates messages within a chosen family of distributions. In the constrained BFE minimization framework, this corresponds to relaxing the strict marginal consistency constraints to moment constraints \cite{heskes2005approximate}. 

Compared to BP, a major problem with EP is that the factor-level beliefs (i.e., approximated marginals involving the arguments of a factor) may be non-integrable. This prevents EP from proceeding, since the computation of means and variances of such beliefs is required. In \cite{zhao2025expectationsexpectationpropagation}, the authors investigated this problem and proposed several methods to address it.

\subsection{Main Contribution}

This paper investigates both BP- and EP-based approaches for estimating the mean and variance of a univariate factored distribution. 

BP cannot be directly applied to approximate these statistics when every factor is, for example, a GMM. To address this issue, we propose Variable Duplication Gaussian Belief Propagation (VDBP) to avoid intractable computations. In multivariate inference problems, BP provides marginal distributions for individual variables, as well as joint marginals for variables appearing together in a factor. Leveraging this property, we construct a multivariate linear measurement model whose underlying marginal posterior (or joint PDF with all but one variable marginalized) is equivalent to the given univariate factored distribution up to a constant factor. We then apply GaBP to iteratively approximate the marginal posterior of the multivariate linear measurement model.

As the name suggests, VDBP converts a univariate factored PDF into a multivariate factored PDF by duplicating one variable for each factor, treating the univariate factored distribution as a multivariate joint PDF in which each factor is proportional to the PDF of an independent random variable. The resulting multivariate factored PDF is interpreted as prior information. The measurement channel and measurement outcome are designed such that the measurement likelihood enforces equality among all duplicated variables. Based on the properties of Dirac delta functions, VDBP constructs the measurement matrix as the transpose of the orthogonal complement of the all-ones vector (i.e., each row of the measurement matrix is orthogonal to the all-ones vector). The measurement outcome and noise are both constructed as all-zero vectors. 
We show analytically that the marginal posterior of the constructed measurement model is equivalent to the original univariate factored PDF. 

To overcome the problem of non-integrable beliefs when EP is used to approximate the statistics of univariate factored distributions, we adopt ideas from recent work \cite{zhao2025expectationsexpectationpropagation}. We investigated a persistent sequential EP algorithm for estimating the univariate factored PDF, in which updates associated with non-integrable beliefs are skipped. 

Furthermore, inspired by \cite{zhao2025expectationsexpectationpropagation}, we propose Analytic Continuation EP (ACEP), in which Gaussian distributions are represented using natural parameters. The technique of analytic continuation allows us to exploit discontinuities in Gaussian functions expressed in natural-parameter form. The key idea of ACEP is to dynamically shrink the projection family in EP to ensure that the belief associated with the next factor is integrable. 

Both EP-based algorithms proposed in this paper operate in a sequential manner. We show that the variable-level beliefs produced by both EP based algorithms are guaranteed to be integrable.

\nocite{schrempf2005optimal}

\section{Introduction to BP \label{sec:IntroBPEP}}
BP is a powerful inference tool that converts a global marginalization problem into multiple local ones. Consider a factored joint PDF with non-negative integrable factors $f_{\alpha}(\bmtheta_{\alpha})$
\beq
    p(\bmtheta)\propto \prod_\alpha f_{\alpha}(\bmtheta_{\alpha}), \label{eq:introBP}
\eeq
where $\bmtheta_{\alpha}$ denotes a subvector of $\bmtheta$ containing a subset of the entries of $\bmtheta$. The marginal PDF $p(\theta_i)=\int p(\bmtheta) d\bmthetab_{i}$, where $\bmthetab_{i}$ denotes all entries of $\bmtheta$ except $\theta_i$, involves the evaluation of a high-dimensional and generally intractable integral. BP approximates such marginal PDFs through an iterative message-passing procedure.

A factor graph corresponding to \eqref{eq:introBP} can be constructed by connecting each factor node to its associated variable nodes whenever the variable appears as an argument of the factor. Figure~\ref{fig:GeneralFG} illustrates an example of such factor graphs.
\begin{figure}[t]
    \centering
    \includegraphics[width=0.48\textwidth]{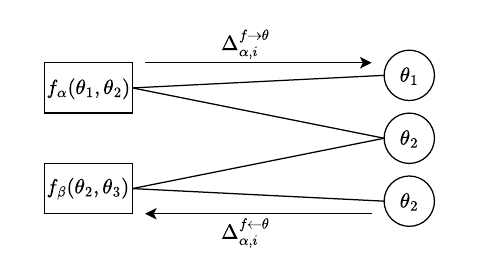}
    \vspace{-4mm}
    \caption{Factor Graph for Joint pdf $p(\theta_1, \theta_2, \theta_3)\propto f_{\alpha}(\theta_1, \theta_2)f_{\beta}(\theta_2, \theta_3)$} 
    \vspace{-2mm}
    \label{fig:GeneralFG}
    \vspace{-4mm}
\end{figure}
Each edge in the graph carries two messages: the message from the factor to the variable node, denoted by $\Delta^{f\to \theta}_{\alpha, i}(\theta_i)$, and the message from the variable to the factor node, denoted by $\Delta^{f\leftarrow\theta}_{\alpha, i}(\theta_i)$.

In essence, BP consists of a set of computational rules that describe how these messages are updated.
Let $N(i)$ denote the set of factor nodes neighboring variable $\theta_i$, and let $N(\alpha)$ denote the set of variable nodes neighboring factor $f_\alpha$. The BP update rules can be summarized as follows.

\subsection{Belief Propagation Rules}
\subsubsection{Factor-to-Variable Message} $\forall \alpha, i$, such that $i\in N(\alpha)$:
\beq
    \Delta^{f\to \theta}_{\alpha, i}(\theta_i)\propto \int f_{\alpha}(\bmtheta_{\alpha})\prod_{j\in N(\alpha)/\{i\}}\Delta^{f\leftarrow \theta}_{\alpha, j}(\theta_j)d\theta_j.
    \label{eq:ICASSP2606}
\eeq
\subsubsection{Variable-to-Factor Message} 
$\forall \alpha, i$, such that $\alpha\in N(i)$: 
\beq
    \Delta^{f\leftarrow \theta}_{\alpha, i}(\theta_i)\propto \!\!\!\!\!\prod_{\beta\in N(i)/\{\alpha\}}\!\!\!\!\!\Delta^{f\to \theta}_{\beta, i}(\theta_i).
    \label{eq:EPv2f}
\eeq
\subsection{Approximation Results of BP}
The BP rules in \eqref{eq:ICASSP2606} and \eqref{eq:EPv2f} are iterated until convergence. The resulting marginal posterior is approximated by variable-level belief
\beq
    p(\theta_i)\simeq b_{\theta_i}(\theta_i):\propto\prod_{\alpha\in N(i)} \Delta^{f\to \theta}_{\alpha, i}(\theta_i).
    \label{eq:ICASSP2607}
\eeq
A remark on BP is that the messages in \eqref{eq:ICASSP2606} and \eqref{eq:EPv2f} do not need to be normalized. However, in the following text, we sometimes refer to the ``mean'' and ``variance'' of these messages, which in fact correspond to those of their normalized versions. In contrast, the approximate marginal distribution (variable-level belief and factor-level belief) $b_{\theta_i}$ computed in \eqref{eq:ICASSP2607} must be normalized to one.

\section{Introduction to EP}

For the same factored joint PDF given in \eqref{eq:introBP}, EP is another message-passing algorithm that can be used to approximate the marginal distributions. Compared to BP, EP not only transforms global inference problems into local ones, but also projects the approximated marginal distributions onto desired families. In this paper, we assume the Gaussian family.

\subsection{Expectation Propagation Rules}

\subsubsection{Belief at Factor Node}
The factor-level belief is computed as
\beq
    b_{f_{\alpha}}(\bmtheta_{\alpha})\propto f_{\alpha}(\bmtheta_{\alpha})\prod_{i\in N(\alpha)}\Delta^{f\leftarrow\theta}_{\alpha, i}(\theta_{i}).
\eeq
Here, the $\propto$ relation indicates that $b_{f_{\alpha}}(\bmtheta_{\alpha})$ must be normalized to one. 
At convergence, $b_{f_{\alpha}}(\bmtheta_{\alpha})$ can be interpreted as an approximation of the  marginal distribution $p(\bmtheta_{\alpha})$.

\subsubsection{Marginalized Factor Belief}
We define the variable-level marginalization of $b_{f_{\alpha}}(\bmtheta_{\alpha})$ as follows: $\forall i, \alpha$, such that $i \in N(\alpha)$, 
\beq
    b_{f_{\alpha}}(\theta_{i})=\int b_{f_{\alpha}}(\bmtheta_{\alpha})  \prod_{j\in N(\alpha)/\{i\}}d\theta_j.
\eeq

\subsubsection{Project to Desired (Gaussian) Family}
Next, we project the marginalized factor belief onto a desired family $\mathcal{F}$ by minimizing the Kullback-Leibler divergence. 
\beq
    \proj\left[b_{f_{\alpha}}\right](\theta_i)=\arg\min_{\bh_{f_{\alpha}}\in\mathcal{F}} \mathrm{KLD}\left[b_{f_{\alpha}}(\theta_i)\|\bh_{f_{\alpha}}(\theta_{i})\right].
    \label{eq:icassp2607t}
\eeq

\subsubsection{Factor-to-Variable Message}
Then, $\forall i, \alpha$, such that $i\in N(\alpha)$, the message from factor $f_{\alpha}$ to variable $\theta_{i}$ is updated by
\beq
    \Delta^{f\to \theta}_{\alpha, i}(\theta_i)\propto \frac{\proj\left[b_{f_{\alpha}}\right](\theta_i)}{\Delta^{f\leftarrow\theta}_{\alpha, i}(\theta_{i})}.
    \label{eq:icassp2608t}
\eeq
Here, $\propto$ indicates equality up to a multiplicative constant.

\subsubsection{Variable Belief}
The variable-level belief is computed as
\beq
    b_{\theta_{i}}(\theta_i)\propto \prod_{\alpha\in N(i)}\!\!\!\!\!\Delta^{f\to \theta}_{\alpha, i}(\theta_i),
\eeq
where $\propto$ indicates that $b_{\theta_{i}}(\theta_i)$ must be normalized to one. At convergence, $b_{\theta_{i}}(\theta_i)$ can be interpreted as the approximation for the marginal distribution $p(\theta_i)$.

\subsubsection{Variable-to-Factor Message}
To close the loop, if $i\in N(\alpha)$, the message from variable $\theta_i$ to factor $f_{\alpha}$ is updated as
\beq
    \Delta^{f\leftarrow \theta}_{\alpha, i}(\theta_i)\propto\frac{b_{\theta_{i}}(\theta_i)}{\Delta^{f\to \theta}_{\alpha, i}(\theta_i)}.
\eeq
Here, $\propto$ denotes equality up to a constant factor.

\section{Problem Formulation}
Consider a univariate factored PDF with non-negative integrable factors $f_{n}(\theta)$
\beq
    p(\theta)\propto \prod_{n=1}^N f_n(\theta),
    \label{eq:firsteq}
\eeq
where the objective of this paper is to develop a low-complexity method for estimating the statistics of $\theta$. The direct computation of integrals involving $p(\theta)$ becomes intractable when each factor $f_{n}$ is a mixture model.

\section{VDBP-Based Approximation for Univariate Factored PDFs}
\subsection{Outline of the VDBP}
Inspired by \cite{rangan2018vectorapproximatemessagepassing}, we duplicate the variable $\theta$ into $N$ new variables $\{\theta_n\}_{n=1}^N$, one for each factor $f_n$, to construct a multivariate joint PDF. To describe the equality relationships among these newly constructed variables, we define a matrix $\bmA\in\mathbb{R}^{M\times N}$ with $M=N-1$, such that $\bmA\cdot\bmone=\bmzero$. This matrix can be constructed by first generating a full-rank random matrix $\bmB\in\mathbb{R}^{M\times N}$ and then setting
\beq
    \bmA=\bmB-\diag\left(\frac{\bmB\cdot\bmone_N}{N}\right)\bmone_{N-1}\bmone_{N}^T.
\eeq
A simple choice for such a matrix is a trimmed Hadamard matrix with the all-ones row removed.

We define the multivariate factored joint PDF as
\beq
    p(\bmtheta)\propto \delta(\bmA\bmtheta)\prod_n f_{n}(\theta_n)\propto \prod_{m}\delta(\bmA_{m,:}\,\bmtheta)\prod_n f_{n}(\theta_n), \label{eq:GMM3}
\eeq
where $\bmA_{m, :}$ denotes the $m$-th row of the matrix $\bmA$. 

\begin{lemma}
\label{lemma:icassp2601}
    $\forall \bmA\in \mathbb{R}^{(N-1)\times N}$ and $\forall n\in\{1, \dots, N\}$, let $\bmAb_n$ denote the matrix obtained by removing the $n$-th column of $\bmA$. If $\rank(\bmA)=N-1$ and $\bmA\cdot \bmone=\bmzero$, then $\rank(\bmAb_n)=N-1$.
\end{lemma}
\begin{proof}
    We prove the lemma by contradiction. Assume that $\rank(\bmAb_n)<N-1$, then $\exists k\in \{1, \dots, N-1\}$ and $\exists \bmb^T \in \mathbb{R}^{1\times (N-2)}$, such that 
    \beq
        \left[\bmAb_{n}\right]_{k, :}=\bmb^T\bmC_{kn},
        \label{eq:icassp2607}
    \eeq
    where $\bmC_{kn}\in \mathbb{R}^{(N-2)\times (N-1)}$ is defined as
    \beq
        \bmC_{kn}=\begin{bmatrix}
            &\left[\bmAb_{n}\right]_{1:(k-1), :}\\
            &\left[\bmAb_{n}\right]_{(k+1):(N-1), :}
        \end{bmatrix}.
        \label{eq:icassp2608}
    \eeq
    The above two equations indicate that the $k$-th row of $\bmAb_{n}$ can be written as a linear combination of the remaining rows in $\bmAb_{n}$.

    On the other hand, due to the assumption $\bmA\cdot\bmone=\bmzero$, we have
    \beq
        \bmAb_{n}\cdot\bmone=-\bma_{n},
        \label{eq:icassp2609}
    \eeq
    where $\bma_{n}$ denotes the $n$-th column of $\bmA$. We can remove the $k$-th rows on both sides of \eqref{eq:icassp2609}:
    \beq
        \bmC_{kn}\cdot\bmone=-\bmd_{kn},
        \label{eq:icassp2610}
    \eeq
    where $\bmd_{kn}$ is a trimmed version of $\bma_{n}$, with the $k$-th entry removed.
    Similarly, by extracting the $k$-th row in \eqref{eq:icassp2609} and we have
    \beq
        \left[\bmAb_{n}\right]_{k, :}\bmone=-[\bmA]_{k, n}.
        \label{eq:icassp2611}
    \eeq

    Right-multiplying both sides of \eqref{eq:icassp2607} by the all-ones vector and using \eqref{eq:icassp2610} and \eqref{eq:icassp2611}, we obtain
    \beq
        \begin{split}
            \left[\bmAb_{n}\right]_{k, :}\cdot\bmone=\bmb^T\bmC_{kn}\cdot\bmone\\
            \Rightarrow [\bmA]_{k, n}= \bmb^T\bmd_{kn}.
        \end{split}
        \label{eq:icassp2612}
    \eeq

    Combining \eqref{eq:icassp2607} and \eqref{eq:icassp2612}, we have
    \beq
    \begin{split}
        \begin{bmatrix}
            [\bmA]_{k, n} & \left[\bmAb_{n}\right]_{k, :}
        \end{bmatrix}=\bmb^T
        \begin{bmatrix}
            \bmd_{kn} & \bmC_{kn}
        \end{bmatrix}\\
        \Rightarrow \begin{bmatrix}
            [\bmA]_{k, n} & \left[\bmAb_{n}\right]_{k, :}
        \end{bmatrix} \bmP=\bmb^T
        \begin{bmatrix}
            \bmd_{kn} & \bmC_{kn}
        \end{bmatrix}\bmP,
    \end{split}
    \label{eq:icassp2613}
    \eeq
    where $\bmP$ is a $N\times N$ permutation matrix:
    \beq
        \bmP=\begin{bmatrix}
            &\bmzero_{n-1}^T &1 &\bmzero_{N-n}^T\\
            &\bmI_{n-1} &\bmzero_{n-1} &\bmO_{(n-1)\times (N-n)}\\
            &\bmO_{(N-n)\times(n-1)} &\bmzero_{N-n} &\bmI_{N-n}
        \end{bmatrix}.
    \eeq
    One can verify that $\begin{bmatrix}
            [\bmA]_{k, n} & \left[\bmAb_{n}\right]_{k, :}
        \end{bmatrix} \bmP$ is the $k$-th row of $\bmA$ and that $\begin{bmatrix}
            \bmd_{kn} & \bmC_{kn}
        \end{bmatrix}\bmP$ is the matrix obtained from $\bmA$ by removing its $k$-th row. 
    Therefore, $\eqref{eq:icassp2613}$ indicates $\rank(\bmA)<N-1$, which contradicts the assumption that $\rank(\bmA)=N-1$.
    
\end{proof}

\begin{theorem}
\label{th:icassp2601}
    $\forall \bmA\in \mathbb{R}^{M\times N}$, where $M=N-1$, if $\rank(\bmA)=M$ and $\bmA\cdot \bmone=\bmzero$, then 
    \beq
        \forall n: \int \delta(\bmA\bmtheta)\prod_{n'\neq n} f_{n'}(\theta_{n'})\mathrm{d}\theta_{n'}\propto \prod_{n'=1}^N f_{n'}(\theta_n).
    \eeq
\end{theorem}
\begin{proof}
    Consider the term inside the Dirac delta function
    \beq
        \delta(\bmA\bmtheta)=\delta(\bmAb_{n}\bmthetab_{n}+\bma_n\theta_n),
    \eeq
    where $\bmAb_{n}$ denotes the $M\times M$ matrix obtained from $\bmA$ by removing its $n$-th column. Similarly, $\bmthetab_n$ is the vector obtained from $\bmtheta$ by removing its $n$-th entry. Furthermore, $\bma_{n}$ denotes the $n$-th column of $\bmA$ and $\theta_n$ denotes the $n$-th entry of $\bmtheta$.
    
    By Lemma~\ref{lemma:icassp2601}, $\bmAb_{n}$ is an invertible matrix. 
    Define $\bmeta=\bmAb_{n}\bmthetab_{n}$, and thus, $\mathrm{d} \bmthetab_{n}=\left|\det\left(\bmAb_{n}^{-1}\right)\right|\mathrm{d}\bmeta$. Therefore, the integral in the theorem can be rewritten as
    \beq
    \begin{split}
        &\int \delta(\bmA\bmtheta)\prod_{n'\neq n} f_{n'}(\theta_{n'})\mathrm{d}\theta_{n'}\\
        &=\frac{1}{\left|\det\left(\bmAb_{n}\right)\right|}\int \mathrm{d}\bmeta\, \delta(\bmeta+\bma_n\theta_n) \!\prod_{n'\neq n}\! f_{n'}\!\left(\left[\bmAb_{n}^{-1}\right]_{k(n'), :}\!\bmeta\right)\\
        &=\frac{1}{\left|\det\left(\bmAb_{n}\right)\right|} \prod_{n'\neq n} f_{n'}\left(-\left[\bmAb_{n}^{-1}\right]_{k(n'), :}\bma_{n}\theta_n\right),
    \end{split}
    \label{eq:icassp2616}
    \eeq
    where $k(n'):\left\{1, \dots, n-1, n+1, \dots, N\right\}\to\{1,\dots, N-1\}$ is an index transformation defined as
    \beq
        k(n')=\begin{cases}
            n' \; & \text{if $1\leq n'<n$}\\
            n'-1 \; &\text{if $n < n'\leq N$}
        \end{cases}.
    \eeq

    On the other hand, from the assumption $\bmA\cdot \bmone=\bmzero$, we have
    \beq
    \begin{split}
        &\bmA\cdot \bmone=\bmzero\\
        &\Rightarrow \bmAb_n\cdot\bmone=-\bma_n\\
        &\Rightarrow \bmone=-\bmAb_{n}^{-1}\bma_n\\
        &\Rightarrow \forall j\in\{1, \dots, N-1\} : \left[\bmAb_{n}^{-1}\right]_{j, :}\bma_{n}=-1.
    \end{split}
    \label{eq:icassp2618}
    \eeq
Substitute \eqref{eq:icassp2618} into the last line of \eqref{eq:icassp2616} yields
\beq
    \int \delta(\bmA\bmtheta)\prod_{n'\neq n} f_{n'}(\theta_{n'})\mathrm{d}\theta_{n'}= \frac{1}{\left|\det\left(\bmAb_{n}\right)\right|} \prod_{n'=1}^N f_{n'}(\theta_n).
\eeq

\end{proof}

Based on Theorem~\ref{th:icassp2601}, we can view \eqref{eq:firsteq} as the marginalized version of \eqref{eq:GMM3}, i.e.,
\beq
    \forall i\in\mathbb{N}\cap[1, N]:\, \prod_{n} f_n(\theta_i)\propto \int \delta(\bmA\bmtheta)\prod_n f_{n}(\theta_n) d\bmthetab_{i}, \label{eq:ICASSP2611}
\eeq
where $\bmthetab_{i}$ denotes an $(N-1)\times 1$ vector obtained by removing the $i$-th entry from $\bmtheta$.

The delta distribution $\delta(\bmA\bmtheta)$ can be interpreted as a noiseless linear measurement model:
\beq
    \bmy=\bmA\bmtheta+\bmv, \label{eq:GMM5}
\eeq
where $\bmy=\bmzero$ and $\bmv\sim\mathcal{N}(\bmzero, \diag(\begin{bmatrix}\epsilon_1 &\dots &\epsilon_N\end{bmatrix}))$, with infinitesimal $\epsilon_n$ for all $n$. The auxiliary noise variance $\epsilon_n$ can serve as a regularization term to prevent zero denominators that may occur due to poor initialization. In our simulations, we set $\epsilon_n=0$. The corresponding factor graph for \eqref{eq:GMM3} and \eqref{eq:GMM5} is illustrated in Figure~\ref{fig:GMMFG}.
\begin{figure}[t]
    \centering
    \includegraphics[width=0.48\textwidth]{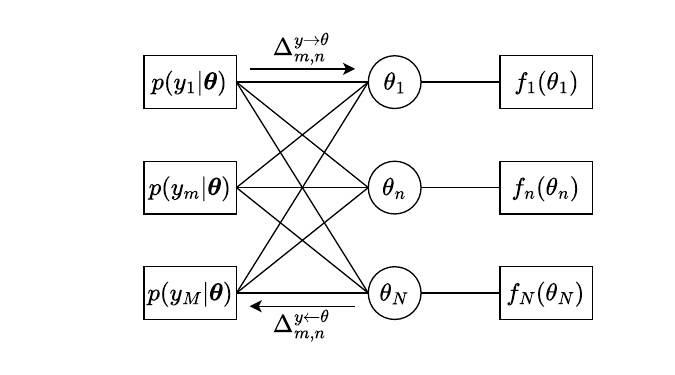}
    \vspace{-2mm}
    \caption{Factor Graph for VDBP} 
    \vspace{-2mm}
    \label{fig:GMMFG}
    \vspace{-4mm}
\end{figure}

Since each factor node $f_i(\theta_i)$ is connected to only one variable $\theta_i$, the message from factor $f_i$ to $\theta_i$ is always $f_i(\theta_i)$.
A tractable BP-based algorithm for estimating the statistics can be obtained by deriving the message update equations for $\Delta^{y\to \theta}_{m, n}$ and $\Delta^{y\leftarrow \theta}_{m, n}$.
These messages can also be interpreted as approximate probabilities: the factor-to-variable message $\Delta^{y\to \theta}_{m, n}$ is proportional to the approximate likelihood $p(y_m|\theta_n)$, while the variable-to-factor message $\Delta^{y\leftarrow\theta}_{m, n}$
is proportional to the approximate posterior $p(\theta_n|\bmyb_{m})$, where $\bmyb_{m}$ denotes all entries of $\bmy$ except the $m$-th one. 

In the following, we apply the BP update rules described in Section~\ref{sec:IntroBPEP}, together with the CLT, to derive these messages.

\subsection{Derivation of the VDBP}

We first derive the factor-to-variable message $\Delta^{y\to \theta}_{m, n}(\theta_n)$. Since BP is an iterative method, we assume that the mean and variance of the variable-to-factor message $\Delta^{y\leftarrow \theta}_{m, n}(\theta_n)$ are given by $\mu^{\theta}_{m, n}=\E_{\Delta^{y\leftarrow \theta}_{m, n}}[\theta_n]$ and $\tau^{\theta}_{m, n}=\var_{\Delta^{y\leftarrow \theta}_{m, n}}[\theta_n]$, respectively.

Define $z_{m}=\sum_{n}a_{mn}\theta_n$. Then the conditional PDF $p(z_{m}|\theta_n, \{\Delta^{y\leftarrow \theta}_{m, n}\}_{n=1}^{N})$ can be approximated as Gaussian according to the CLT
\begin{align}
    &p(z_{m}|\theta_{n}, \{\Delta^{y\leftarrow \theta}_{m, n}\}_{n=1}^{N})\nonumber\\
    &=\int \delta(z_m-a_{mn}\theta_n-\sum_{n'\neq n}a_{mn'}\theta_{n'}) \prod_{n'\neq n}\Delta^{y\leftarrow\theta}_{m, n'}(\theta_{n'})d\theta_{n'} \nonumber\\
    &\simeq \mathcal{N}(z_m|a_{mn}\theta_n+\mu^{p}_{n, m}, \tau^{p}_{n, m}), \label{eq:ICASSP2613}
\end{align}
where $\{\Delta^{y\leftarrow \theta}_{m, n}\}_{n=1}^{N}$ denotes the collection of all variable-to-factor messages. In the last line of \eqref{eq:ICASSP2613}, the CLT is applied to approximate the interference term $\sum_{n'\neq n}a_{mn'}\theta_{n'}$ as Gaussian with mean and variance $\mu^{p}_{n, m}$, $\tau^{p}_{n, m}$:
\beq
    \begin{split}
        \mu^{p}_{n, m}=\sum_{n'\neq n}a_{mn'}\mu^{\theta}_{m, n'}\\
        \tau^{p}_{n, m}=\sum_{n'\neq n}|a_{mn'}|^2\tau^{\theta}_{m, n'}.
        \label{eq:ICASSP2614}
    \end{split}
\eeq

Based on \eqref{eq:ICASSP2606}, we combine $p(z_m|\theta_n, \{\Delta^{y\leftarrow \theta}_{m, n}\}_{n=1}^{N})$ with the measurement model $p(y_m|z_m)=\mathcal{N}(y_m|z_m, \epsilon_m)$ and integrate over $z_m$ to obtain the factor-to-variable message $\Delta^{y\to \theta}_{m, n}(\theta_n)$, which corresponds to the approximate likelihood $p(y_m|\theta_n)$
\beq
    \Delta^{y\to \theta}_{m, n}(\theta_n)\propto \mathcal{N}\left(\theta_n\left|\frac{y_m-\mu^{p}_{n, m}}{a_{mn}}, \frac{\tau^{p}_{n, m}+\epsilon_m}{|a_{mn}|^2}\right.\right).
    \label{eq:ICASSP2615}
\eeq

We now derive the variable-to-factor message $\Delta^{y\leftarrow \theta}_{m, n}$ using \eqref{eq:EPv2f}:
\beq
\begin{split}
\Delta^{y\leftarrow \theta}_{m, n}(\theta_n)\propto f_{n}(\theta_n)\prod_{j\neq m}p (y_j|\theta_n)\\
=f_{n}(\theta_n)\mathcal{N}(\theta_n|\mu^{r}_{m, n}, \tau^{r}_{m, n}),
\end{split}
\label{eq:icassp2544}
\eeq
where $\mu^r_{m, n}$ and $\tau^r_{m, n}$ follow from the Gaussian reproduction property \cite{8496782}:
\beq
    \begin{split}
        &\tau^{r}_{m, n}=\left(\sum_{j\neq m}\frac{|a_{jn}|^2}{\tau^{p}_{n, j}+\epsilon_j}\right)^{-1}\\
        &\mu^{r}_{m, n}=\tau^{r}_{m, n}\left(\sum_{j\neq m}\frac{a_{jn}\left(y_j-\mu^{p}_{n, j}\right)}{\tau^{p}_{n, j}+\epsilon_j}\right).
    \end{split}
\eeq
From \eqref{eq:ICASSP2613}–\eqref{eq:ICASSP2614}, we observe that under the CLT approximation only the first two moments of $\Delta^{y\leftarrow \theta}_{m, n}$ are required:
\beq
    \begin{split}
        \mu^{\theta}_{m, n}=\E_{\Delta^{y\leftarrow \theta}_{m, n}}[\theta_n]\\
        \tau^{\theta}_{m, n}=\var_{\Delta^{y\leftarrow \theta}_{m, n}}[\theta_n],
    \end{split}
    \label{eq:ICASSP2618}
\eeq
where $\Delta^{y\leftarrow \theta}_{m, n}$ is defined in \eqref{eq:icassp2544}. This concludes our discussion for a general factor $f_{n}$ in \eqref{eq:firsteq}. Next, we consider a representative case in which $f_{n}$ follows a GMM that arises in bilinear joint estimation and detection problems.

\subsubsection{GMM Factors}
Assume that the factors in \eqref{eq:firsteq} have the following structure:
\beq
    f_n(\theta_n)\propto \sum_{s_n} p_{s_n} \mathcal{N}(\theta_n|\mu_{s_n}, \tau_{s_n}). \label{eq:ICASSP2619}
\eeq
The only part of VDBP that depends on the factor structure is the update of the mean and variance of $\Delta^{y\leftarrow\theta}_{m, n}$ in \eqref{eq:ICASSP2618}.
Substituting \eqref{eq:ICASSP2619} into \eqref{eq:icassp2544} and applying the Gaussian Reproduction Lemma \cite{8496782}, we obtain
\beq
\begin{split}
    &\Delta^{y\leftarrow\theta}_{m, n}(\theta_n)
    \propto \sum_{s_n}\pi_n(s_n|\mu^{r}_{m, n}, \tau^{r}_{m, n})  \\
    &\cdot \mathcal{N}(\theta_n|\mu^{x|s}_{m, n}(s_n), \tau^{x|s}_{m, n}(s_n)),
\end{split}
\label{eq:ICASSP2620}
\eeq
where $\pi_n(s_n|\mu^{r}_{m, n}, \tau^{r}_{m, n})$ is the normalized categorical distribution defined as
\beq
    \pi_n(s_n|\mu^{r}_{m, n}, \tau^{r}_{m, n})\propto p_{s_n} \mathcal{N}(0|\mu_{s_n}-\mu^{r}_{m, n}, \tau^{r}_{m, n}+\tau_{s_n}).
    \label{eq:icassp2638trgdfg}
\eeq
The corresponding conditional mean and variance are given by
\beq
    \begin{split}
        &\mu^{x|s}_{m, n}(s_n)=\frac{\tau^{r}_{m, n}\mu_{s_n}+\tau_{s_n}\mu^{r}_{m, n}}{\tau^{r}_{m, n}+\tau_{s_n}}\\
        &\tau^{x|s}_{m, n}(s_n)=\frac{\tau^{r}_{m, n}\tau_{s_n}}{\tau^{r}_{m, n}+\tau_{s_n}}.
    \end{split}
\eeq
Therefore, the message $\Delta^{y\leftarrow \theta}_{m, n}(\theta_n)$ in \eqref{eq:ICASSP2620} is itself a GMM with the same number of components as the factor $f_{n}$.
By exchanging the order of integration and summation, the mean of $\Delta^{y\leftarrow\theta}_{m, n}(\theta_n)$ can be computed as
\beq
    \mu^{\theta}_{m, n}=\E_{\pi_n(s_n|\mu^{r}_{m, n}, \tau^{r}_{m, n})}\left[\mu^{x|s}_{m, n}(s_n)\right],
    \label{eq:icassp2548}
\eeq
while its variance is given by
\beq
    \tau^{\theta}_{m, n}=\E_{\pi_n(s_n|\mu^{r}_{m, n}, \tau^{r}_{m, n})}\left[\tau^{x|s}_{m, n}(s_n)+\mu^{x|s}_{m, n}(s_n)^2\right]-(\mu^{\theta}_{m, n})^2.
    \label{eq:icassp2549}
\eeq

\subsection{Approximated Statistics}
Recall the relation in \eqref{eq:ICASSP2611}, which indicates that the  univariate factorized PDF in \eqref{eq:firsteq} can be approximated by the marginal distribution of \eqref{eq:GMM3}. According to the BP rule in \eqref{eq:ICASSP2607}, the marginal distribution of a variable can be approximated as the product of all incoming messages to that variable. Since the message $\Delta^{y\to \theta}_{m, n}$ in \eqref{eq:ICASSP2615} follows a Gaussian distribution, we define the extrinsic mean and variance $\mu^{r}_{n}$ and $\tau^{r}_{n}$ such that
\beq
    \mathcal{N}(\theta_n| \mu^{r}_{n}, \tau^{r}_{n}):\propto \prod_{m} \Delta^{y\to \theta}_{m, n}(\theta_n). \label{eq:ICASSP2625}
\eeq
As a result, the extrinsic mean and variance are
\beq
\begin{split}
    &\tau^{r}_{n}=\left(\sum_{m}\frac{|a_{mn}|^2}{\tau^{p}_{n, m}+\epsilon_m}\right)^{-1}\\
    &\mu^{r}_{n}=\tau^{r}_{n}\left(\sum_{m}\frac{a_{mn}\left(y_m-\mu^{p}_{n, m}\right)}{\tau^{p}_{n, m}+\epsilon_m}\right).    
\end{split}
\eeq
By combining the extrinsic distribution in \eqref{eq:ICASSP2625} with the factor $f_{n}(\theta_n)$, we obtain the approximate marginal of \eqref{eq:GMM3} as
\beq
    b_{\theta_n}(\theta_n):\propto f_{n}(\theta_n)\mathcal{N}(\theta_n|\mu^{r}_{n}, \tau^{r}_{n}). \label{eq:ICASSP2627}
\eeq
According to \eqref{eq:ICASSP2611}, \eqref{eq:ICASSP2627} can be regarded as an approximate form of \eqref{eq:firsteq}. Without loss of generality, we restrict our attention to the approximate mean and variance.

Each $\theta_n$ yields its own belief of the marginal mean and variance, denoted by $\mu^{\theta}_{n}$ and $\tau^{\theta}_{n}$:
\beq
    \begin{split}
        \mu^{\theta}_{n}=\E_{b_{\theta_n}}[\theta_n]\\
        \tau^{\theta}_{n}=\var_{b_{\theta_n}}[\theta_n].
    \end{split}
    \label{eq:ICASSP2628}
\eeq
Finally, we combine these beliefs to obtain the approximate mean and variance of $p(\theta)$ in \eqref{eq:firsteq}, denoted by $\mu^{\thetah}$ and $\tau^{\thetah}$:
\beq
\begin{split}
    &\mu^{\thetah}=\left(\sum_{n}\frac{1}{\tau^{\theta}_{n}} \right)^{-1}\left(\sum_{n}\frac{\mu^{\theta}_{n}}{\tau^{\theta}_{n}}\right)\\
    &\tau^{\thetah}=\min(\tau^{\theta}_{1},\dots, \tau^{\theta}_{N}).
\end{split}
\eeq
The computational complexity of the proposed algorithm depends on the number of nonzero elements in $\bmA$. When a trimmed Hadamard matrix is used, the overall computational complexity is $O(N^2)$. The VDBP algorithm is summarized in Algorithm \ref{algo:wsa2501}.

\floatstyle{spaceruled}
\restylefloat{algorithm}
\begin{algorithm}[t]
\caption{VDBP for GMM Factors}\label{algo:wsa2501}
\begin{algorithmic}[1]
\State Initialize $\forall m, n: \mu^{\theta}_{m, n}=0, \tau^{\theta}_{m, n}=1$
\Repeat
\State $\forall m:\, \mu^{p}_{m}=\sum_{n} a_{mn} \mu^{\theta}_{m, n}$
\State $\forall m:\, \tau^{p}_{m}=\sum_{n} |a_{mn}|^2\tau^{\theta}_{m, n}$
\State $\forall n, m:\, \mu^{p}_{n, m}=\mu^{p}_{m}-a_{mn} \mu^{\theta}_{m, n}$
\State $\forall n, m:\, \tau^{p}_{n, m}=\tau^{p}_{m}-|a_{mn}|^2\tau^{\theta}_{m, n}$
\State $\forall n:\, \tau^{r}_{n}=\left(\sum_{m}\frac{|a_{mn}|^2}{\tau^{p}_{n, m}+\epsilon_{m}}\right)^{-1}$
\State $\forall n:\, \mu^{r}_{n}=\tau^{r}_{n}\left(\sum_{m}\frac{a_{mn}^*\left(y_m-\mu^{p}_{n, m}\right)}{\tau^{p}_{n, m}+\epsilon_m}\right)$
\State $\forall m, n:\, \tau^{r}_{m, n}=\left(\tau^{r\,-1}_{n}-\frac{|a_{mn}|^2}{\tau^{p}_{n, m}+\epsilon_{m}}\right)^{-1}$
\State $\forall m, n:\, \mu^{r}_{m, n}=\tau^{r}_{m, n}\left(\frac{\mu^{r}_{n}}{\tau^{r}_{n}}-\frac{a_{mn}^*\left(y_m-\mu^{p}_{n, m}\right)}{\tau^{p}_{n, m}+\epsilon_m}\right)$
\State $\forall m, n:\, \text{Compute normalized }\pi_n(s_n)$ via \eqref{eq:icassp2638trgdfg}
\State $\forall m, n:$ Compute $\mu^{\theta}_{m, n}$, $\tau^{\theta}_{m, n}$ via \eqref{eq:icassp2548}, \eqref{eq:icassp2549}
\Until{Convergence}
\end{algorithmic}
\end{algorithm}

\section{Preliminaries for Expectation Propagation}
In EP, a message does not necessarily correspond to a normalizable function and may even be non-integrable.
Let $\mathcal{UN}(\theta|\mu, \tau)$ denote an unnormalized Gaussian (ignoring the $\sqrt{2\pi\tau}$ normalization factor), 
\beq
    \mathcal{UN}(\theta|\mu, \tau)=\exp{-\frac{\left(\theta-\mu\right)^2}{2\tau}}.
    \label{eq:icassp2646}
\eeq

We define a Gaussian distribution in natural-parameter form as 
\beq
    \mathcal{M}(\theta|\nu, \xi)=\frac{\sqrt{\xi}}{\sqrt{2\pi}}\exp{\left(-\frac{\xi}{2}\theta^{2}+\nu\theta-\frac{\nu^2}{2\xi}\right)}.
\eeq
Similarly, the unnormalized Gaussian in natural-parameter form is defined as
\beq
    \mathcal{UM}(\theta|\nu, \xi)=\exp{\left(-\frac{\xi}{2}\theta^{2}+\nu\theta-\frac{\nu^2}{2\xi}\right)}.
\eeq
One can verify that if $\mathcal{UN}(\theta|\mu, \tau)=\mathcal{UM}(\theta|\nu, \xi)$, or if $\mathcal{N}(\theta|\mu, \tau)=\mathcal{M}(\theta|\nu, \xi)$, then
\beq
\begin{split}
    \nu=\frac{\mu}{\tau};\\
    \xi=\frac{1}{\tau}.
\end{split}
\label{eq:Icassp2649}
\eeq

We can then obtain the unnormalized Gaussian reproduction lemma. 
For all $\mu\in \mathbb{R}$ and $\tau\in \mathbb{R}/\{0\}$, 
\beq
\begin{split}
    &\mathcal{UN}(\theta|\mu_1, \tau_1)\mathcal{UN}(\theta|\mu_2, \tau_2)\\
    =&\mathcal{UN}(0|\mu_1-\mu_2, \tau_1+\tau_2)\mathcal{UN}\left(\theta\left|\frac{\tau_2\mu_1+\tau_1\mu_2}{\tau_1+\tau_2}, \frac{\tau_1\tau_2}{\tau_1+\tau_2}\right.\right).
\end{split}
\eeq
Similarly, apply the Gaussian reproduction lemma to unnormalized Gaussians in natural-parameter form, and using analytic continuation to cover the case $\xi=0$ (i.e., $\tau\to \infty$), we have  $\forall \nu, \xi\in \mathbb{R}$,
\beq
    \begin{split}
    &\mathcal{UM}(\theta|\nu_1, \xi_1)\mathcal{UM}(\theta|\nu_2, \xi_2)\!=\!\mathcal{UN}(\theta|\mu_1, \tau_1)\mathcal{UN}(\theta|\mu_2, \tau_2)\\
    =&\mathcal{UM}\left(0\left|\frac{\xi_2\nu_1\!-\!\xi_1\nu_2}{\xi_1+\xi_2}, \frac{\xi_1\xi_2}{\xi_1\!+\!\xi_2}\right.\right)\mathcal{UM}\left(\theta\left|\nu_1\!+\!\nu_2, \xi_1+\xi_2\right.\right),
    \end{split}
    \label{eq:ICASSP2632}
\eeq
where 
\beq
\begin{split}
    \nu_1=\frac{\mu_1}{\tau_1}, \, \nu_2=\frac{\mu_2}{\tau_2};\\
    \xi_1=\frac{1}{\tau_1}, \, \xi_2=\frac{1}{\tau_2}.
\end{split}
\eeq

\subsection{Factor Graph of the Given Problem in EP}
Now we investigate EP-based algorithms.  Consider the univariate factored PDF given in \eqref{eq:firsteq}. We define the message from the $n$-th factor to $\theta$ as $\Delta^{f\to \theta}_{n}(\theta)$ and the message from $\theta$ to the $n$-th factor as $\Delta^{f\leftarrow \theta}_{n}(\theta)$. The corresponding factor graph is shown in Fig. \ref{fig:GMMEP}.
\begin{figure}[t]
    \centering
    \includegraphics[width=0.48\textwidth]{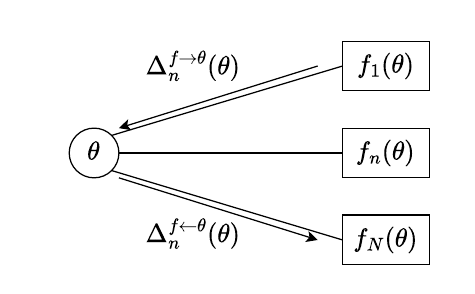}
    \vspace{-2mm}
    \caption{Factor Graph for EP} 
    \vspace{-2mm}
    \label{fig:GMMEP}
    \vspace{-4mm}
\end{figure}

\section{Persistent EP}
In EP, the messages are not required to be integrable.
Unlike the Gaussian distribution, the unnormalized Gaussian defined in \eqref{eq:icassp2646} allows negative values of $\tau$. However, the point $\tau=0$ is an essential discontinuity of \eqref{eq:icassp2646}.
Persistent EP operates under a sequential update order. Specifically, the messages to and from one factor are updated at a time, and the latest messages are immediately used when updating the messages associated with the next factor. In the following, we focus on deriving the messages to and from factor $f_{n}$.

Persistent EP does not actively prevent the formation of non-integrable beliefs. Instead, it handles non-integrable beliefs in a passive manner: if the belief at the factor $f_{n}(\theta)$ is non-integrable, the update of the message $\Delta^{f\to \theta}_{n}(\theta)$ is skipped.

Based on the EP rules, the message from variable to factor $f_{n}(\theta)$ can be computed as
\beq
    \Delta^{f\leftarrow \theta}_{n}(\theta)=\mathcal{UN}(\theta|\mu^{f\leftarrow \theta}_{n}, \tau^{f\leftarrow \theta}_{n})\propto \prod_{n'\neq n} \Delta^{f\to \theta}_{n'}(\theta), 
    \label{eq:icassp2647}
\eeq
where 
\beq
\begin{split}
    \tau^{f\leftarrow \theta}_{n}=\left(\sum_{n'\neq n} \frac{1}{\tau^{f\to \theta}_{n'}}\right)^{-1}\\
    \mu^{f\leftarrow \theta}_{n}=\tau^{f\leftarrow \theta}_{n}\left(\sum_{n'\neq n}\frac{\mu^{f\to \theta}_{n'}}{\tau^{f\to \theta}_{n'}}\right).
\end{split}
\label{eq:icassp2654t}
\eeq
Unlike standard EP algorithm, a checking mechanism is introduced before computing the belief at factor $f_{n}$. Specifically, we evaluate whether the integral 
\[
    \int f_{n}(\theta)\Delta^{f\leftarrow \theta}_{n}(\theta)\mathrm{d}\theta
\]
is finite. The check passes if the integral exist in $\mathbb{R}$.

If the check fails, the update of $\Delta^{f\to \theta}_{n}(\theta)$ is skipped in the current iteration, and the algorithm proceeds to the next factor $f_{k}$, where $k=(n\mod N)+1$, and repeats from \eqref{eq:icassp2647} with index $n$ replaced by $k$.

If the check passes, the belief at $f_{n}$ can be computed as
\beq
    b_{f_{n}}(\theta)\propto f_{n}(\theta)\Delta^{f\leftarrow \theta}_{n}(\theta), 
    \label{eq:icassp2649}
\eeq
which is guaranteed to have a finite integral value. The corresponding belief mean and variance are
\beq
\begin{split}
    \mu^{b_{f}}_{n}=\E_{b_{f_{n}}}[\theta],\\
    \tau^{b_{f}}_{n}=\var_{b_{f_{n}}}[\theta].
\end{split}
\label{eq:icassp26566}
\eeq
According to the EP update rule, the message from factor $f_{n}$ to $\theta$ is
\beq
    \Delta^{f\to \theta}_{n}(\theta)=\mathcal{UN}(\theta|\mu^{f\to \theta}_{n}, \tau^{f\to \theta}_{n})\propto\frac{\mathcal{UN}(\theta| \mu^{b_{f}}_{n}, \tau^{b_{f}}_{n})}{\Delta^{f\leftarrow\theta}_{n}(\theta)},
    \label{eq:icassp2651}
\eeq
where
\beq
    \begin{split}
        &\tau^{f\to \theta}_{n}=\left(\frac{1}{\tau^{b_{f}}_{n}}-\frac{1}{\tau^{f\leftarrow \theta}_{n}}\right)^{-1}\\
        &\mu^{f\to \theta}_{n}=\tau^{f\to \theta}_{n}\left(\frac{\mu^{b_f}_{n}}{\tau^{b_{f}}_{n}}-\frac{\mu^{f\leftarrow \theta}_{n}}{\tau^{f\leftarrow \theta}_{n}}\right).
    \end{split}
    \label{eq:icassp2658}
\eeq
After the final iteration, the Gaussian approximation of the univariate factored PDF is the belief at variable $\theta$:
\beq
    b_{\theta}(\theta)\propto \prod_{n}\Delta^{f\to \theta}_{n}(\theta).
    \label{eq:icassp2653}
\eeq
Therefore, the mean and variance of $b_{\theta}(\theta)$ are the approximated mean and variance of the univariate factored PDF:
\beq
\begin{split}
    \tau^{b_{\theta}}=\left(\sum_{n} \frac{1}{\tau^{f\to \theta}_{n}}\right)^{-1}\\
    \mu^{b_{\theta}}=\tau^{b_{\theta}}\left(\sum_{n}\frac{\mu^{f\to \theta}_{n}}{\tau^{f\to \theta}_{n}}\right).
\end{split}
\label{eq:icassp2654}
\eeq
The above procedure is then repeated for the next factor $f_{k}$, where $k=(n\mod N)+1$. Since a sequential update strategy is used, $b_{\theta}(\theta)$ is updated by \eqref{eq:icassp2653} whenever a factor-to-variable message is updated. Thus, \eqref{eq:icassp2653} serves as the definition for $b_{\theta}(\theta)$ throughout the algorithm.
\begin{proposition}
\label{prop:icassp2601}
    Assume that $b_{\theta}(\theta)$ is defined as in \eqref{eq:icassp2653}. If the algorithm in \eqref{eq:icassp2647}-\eqref{eq:icassp2654} is initialized with messages such that $b_{\theta}(\theta)$ is integrable, then the integral of $b_{\theta}(\theta)$ remains finite throughout all iterations.
\end{proposition}
\begin{proof}
    We prove the proposition by mathematical induction over the sub-iterations.
    The $i$-th sub-iteration is defined as the update in which only the messages to and from the $n$-th factor are considered, where $n=[(i-1)\mod N]+1$ . An iteration therefore consists of $N$ sub-iterations. 

    By assumption, the belief $b_{\theta}(\theta)$ is initialized to be integrable. 

    For clarity, we denote $\forall k=\{1, \dots, N\}: \Delta^{f\leftarrow \theta}_{k; (i)}(\theta)$ and $\Delta^{f\to \theta}_{k; (i)}(\theta)$ as the messages to and from factor $f_{k}(\theta)$ at the end of sub-iteration $i$. Similarly, let $b_{\theta; (i)}(\theta)$ denote the variable-level belief at the end of sub-iteration $i$. We also denote $\mu^{b_{f}}_{n;(i)}$ and $\tau^{b_{f}}_{n;(i)}$ as the mean and variance of the factor-level belief $b_{f_{n}}(\theta)$ at the end of iteration $i$.
    
    Assume that at the end of sub-iteration $i$, the belief $b_{\theta;(i)}(\theta)$ is integrable. At sub-iteration $i+1$, the messages to and from factor $f_{m}(\theta)$ are updated, where $m=(i\mod N)+1$. If the factor-level belief computed in \eqref{eq:icassp2649}
    is not integrable, i.e., 
    \beq
        \int f_{m}(\theta)\Delta^{f\leftarrow \theta}_{m; (i+1)}(\theta)\mathrm{d}\theta=\infty, 
    \eeq
    then the update of $\Delta^{f\to \theta}_{m}(\theta)$ is skipped and 
    \beq
    \Delta^{f\to \theta}_{m; (i+1)}(\theta)=\Delta^{f\to \theta}_{m; (i)}(\theta).
    \eeq
    Consequently, 
    \beq
        b_{\theta; (i+1)}(\theta)=b_{\theta; (i)}(\theta),
    \eeq
    which remains integrable. 

    If instead the factor-level belief computed in \eqref{eq:icassp2649} is integrable, i.e., 
    \beq
        \int f_{m}(\theta)\Delta^{f\leftarrow \theta}_{m; (i+1)}(\theta)\mathrm{d}\theta<\infty, 
    \eeq
    then by \eqref{eq:icassp2651}, we have
    \beq
        \Delta^{f\to \theta}_{m; (i+1)}(\theta)\Delta^{f\leftarrow \theta}_{m; (i+1)}(\theta)=\mathcal{UN}(\theta| \mu^{b_{f}}_{m; (i+1)}, \tau^{b_{f}}_{m;(i+1)}),
        \label{eq:icassp2657}
    \eeq
    where $\tau^{b_{f}}_{m;(i+1)}\geq 0$. 
    Since $\tau^{b_{f}}_{m;(i+1)}\geq 0$, we can verify \eqref{eq:icassp2657} has finite integral (if $\tau^{b_{f}}_{m;(i+1)}= 0$, then \eqref{eq:icassp2657} is Dirac delta function). 

    Finally, comparing \eqref{eq:icassp2647} and \eqref{eq:icassp2653}, we have
    \beq
        b_{\theta;(i+1)}(\theta)=\Delta^{f\to \theta}_{m; (i+1)}(\theta)\Delta^{f\leftarrow \theta}_{m; (i+1)}(\theta).
    \eeq
    Thus, $b_{\theta;(i+1)}(\theta)$ is also integrable.

    By induction, $b_{\theta}(\theta)$ has a finite integral at every sub-iteration.
\end{proof}

The persistent EP algorithm for GMM factors is concluded in Algorithm~\ref{algo:icassp2602}. Due to the sequential update strategy, line \ref{PEPline4} of Algorithm \ref{algo:icassp2602} has computational complexity $O(1)$, since only a single term corresponding to the most recently updated message $\Delta^{f\to \theta}_{k}$ is updated when updating $b_{\theta}(\theta)$, where $k=[(n-2)\mod N]+1$.

\floatstyle{spaceruled}
\restylefloat{algorithm}
\begin{algorithm}[t]
\caption{Persistent EP for GMM Factors}\label{algo:icassp2602}
\begin{algorithmic}[1]
\State Initialize $\forall n: \mu^{f\to \theta}_{n}=0, \tau^{f\to \theta}_{n}=1$
\Repeat
\State [For every $n\in [1,\dots, N]$, repeat the following]
\State Update $\tau^{b_{\theta}}$ and $\mu^{b_{\theta}}$ by \eqref{eq:icassp2654} \label{PEPline4}
\State $\tau^{f\leftarrow \theta}_{n}=\frac{\tau^{f\to\theta}_{n}\tau^{b_{\theta}}}{\tau^{f\to\theta}_{n}-\tau^{b_{\theta}}}$
\State $\mu^{f\leftarrow \theta}_{n}=\frac{\tau^{f\to\theta}_{n}\mu^{b_{\theta}}-\tau^{b_{\theta}}\mu^{f\to\theta}_{n}}{\tau^{f\to\theta}_{n}-\tau^{b_{\theta}}}$
\If {$b_{f_{n}}(\theta)$ in \eqref{eq:icassp2649} is integrable}
\State Update $\mu^{b_{f}}_{n}$ and $\tau^{b_{f}}_{n}$ by \eqref{eq:icassp2649} - \eqref{eq:icassp26566}
\State Update $\tau^{f\to \theta}_{n}$ and $\mu^{f\to \theta}_n$ by \eqref{eq:icassp2658}
\EndIf

\Until{Convergence}
\State Output $\mu^{b_{\theta}}$ and $\tau^{b_{\theta}}$ as approximated mean and variance
\end{algorithmic}
\end{algorithm}

\section{Analytic Continuation EP}

To avoid non-integrable beliefs that prevent the algorithm from progressing, we adopt the technique of analytic continuation proposed in \cite{zhao2025expectationsexpectationpropagation}. While persistent EP handles non-integrable beliefs in a passive manner, Analytic Continuation EP (ACEP) actively prevents non-integrable factor-level belief at the subsequent factor. By expressing Gaussian distributions in terms of their natural parameters as in \eqref{eq:Icassp2649}, we have a removable discontinuity at $\xi=0$ in $\mathcal{UM}$ if we discard the constant term $-\frac{\nu^2}{\xi}$, and the essential discontinuity $\tau=0$ in $\mathcal{UN}$ becomes $\xi\to \pm \infty$. Consequently, it is more convenient in this section to represent Gaussian distributions using the natural parameters $(\nu, \xi)$ rather than the mean–variance pair $(\mu, \tau)$.  

Based on EP, the message from $\theta$ to the $n$-th factor is updated as
\beq
    \Delta^{f\leftarrow \theta}_{n}(\theta)=\mathcal{UM}(\theta|\nu^{f\leftarrow \theta}_{n}, \xi^{f\leftarrow \theta}_{n})\propto \prod_{n'\neq n} \Delta^{f\to \theta}_{n'}(\theta),
    \label{eq:icassp2665}
\eeq
where 
\beq
    \begin{split}
        \nu^{f\leftarrow\theta}_{n}=\sum_{n'\neq n}\nu^{f\to\theta}_{n'};\\
        \xi^{f\leftarrow\theta}_{n}=\sum_{n'\neq n}\xi^{f\to\theta}_{n'}.
    \end{split}
    \label{eq:icassp2666}
\eeq
The belief at $f_{n}$ is then computed as
\beq
    b_{f_{n}}(\theta)\propto f_{n}(\theta)\Delta^{f\leftarrow \theta}_{n}(\theta).
    \label{eq:ICASSP2636true}
\eeq
When the Gaussian family is chosen as the projection family in EP, minimizing the Kullback–Leibler divergence is equivalent to matching the mean and variance. Accordingly, we define the approximated posterior mean and variance of the belief at factor $f_n$ as
\beq
\begin{split}
    \mu^{b_{f}}_{n}=\E_{b_{f_{n}}}[\theta],\\
    \tau^{b_{f}}_{n}=\var_{b_{f_{n}}}[\theta].
\end{split}
\label{eq:ICASSP2637}
\eeq
Following \eqref{eq:ICASSP2637}, the corresponding natural parameters are defined as
\beq
    \begin{split}
        \nu^{b_{f}}_{n}=\frac{\mu^{b_{f}}_{n}}{\tau^{b_{f}}_{n}};\\
        \xi^{b_{f}}_{n}=\frac{1}{\tau^{b_{f}}_{n}}.
    \end{split}
\eeq
Inspired by \cite{zhao2025expectationsexpectationpropagation}, to prevent non-integrable beliefs from blocking further updates, it suffices to ensure that the belief at the next factor $f_{k}(\theta)$ remains integrable, where $k=(n \mod N)+1$. 

Due to the following two observations:
\begin{itemize}
    \item $\left(\nu^{f\to \theta}_n, \xi^{f\to \theta}_n\right)$ and $\left(\nu^{f\leftarrow \theta}_k, \xi^{f\leftarrow \theta}_k\right)$ are linearly related through \eqref{eq:icassp2666};
    \item no other messages are updated between the updates of $\Delta^{f\to \theta}_{n}(\theta)$ and $\Delta^{f\leftarrow \theta}_{k}(\theta)$,
\end{itemize}
and hence, there exists a bijection between $\left(\nu^{f\to \theta}_n, \xi^{f\to \theta}_n\right)$ and $\left(\nu^{f\leftarrow \theta}_k, \xi^{f\leftarrow \theta}_k\right)$. We therefore define a domain $\Omega_{n}$ for $\left(\nu^{f\to \theta}_n, \xi^{f\to \theta}_n\right)$ such that, if $\left(\nu^{f\to \theta}_n, \xi^{f\to \theta}_n\right)\in \Omega_{n}$ the belief update at the next factor node
\beq
    b_{f_{k}}(\theta)\propto f_{k}(\theta)\Delta^{f\leftarrow \theta}_{k}(\theta)
\eeq
is integrable, where  $k=(n \mod N)+1$. 

Based on the EP rules in \eqref{eq:icassp2607t} and \eqref{eq:icassp2608t}, ACEP merges the projection and message-update steps into a direct update of the factor-to-variable message.

From \eqref{eq:icassp2608t}, the projected belief admits the form
\beq
\begin{split}
    \bh_{f_{n}}(\theta)&=\mathcal{M}\left[\theta\left| \nu^{\bh_{f}}_n\left(\nu^{f\to \theta}_n, \xi^{f\to \theta}_n\right), \xi^{\bh_{f}}_{n}\left(\nu^{f\to \theta}_n, \xi^{f\to \theta}_n\right)\right.\right]\\
    &\propto \Delta^{f\leftarrow \theta}_{n}(\theta) \Delta^{f\to \theta}_{n}(\theta),
\end{split}
\label{eq:icassp2671}
\eeq
with
\beq
\begin{split}
    \nu^{\bh_{f}}_n\left(\nu^{f\to \theta}_n, \xi^{f\to \theta}_n\right)=\nu^{f\to \theta}_n+\nu^{f\leftarrow \theta}_{n}\\
    \xi^{\bh_{f}}_n\left(\nu^{f\to \theta}_n, \xi^{f\to \theta}_n\right)=\xi^{f\to \theta}_n +\xi^{f\leftarrow \theta}_{n}.
\end{split}
\label{eq:icassp2672}
\eeq
Since ACEP is iterative, the parameters $\left(\nu^{f\leftarrow \theta}_{n}, \xi^{f\leftarrow \theta}_{n}\right)$ are treated as constants during the update of $\Delta^{f\to \theta}_{n}(\theta)$. Substituting \eqref{eq:icassp2671} and \eqref{eq:icassp2672} into the EP rule \eqref{eq:icassp2607t} as the second KLD parameter, and incorporating the domain constraint $\Omega_{n}$, yields
\beq
    \begin{split}
        \left(\nu^{f\to \theta}_n, \xi^{f\to \theta}_n\right)=\arg\!\!\!\!\!\!\!\min_{\left(\nu^{f\to \theta}_n, \xi^{f\to \theta}_n\right)\in \Omega_{n}} \mathrm{KLD}\left[b_{f_{n}}(\theta)\|\bh_{f_{n}}(\theta)\right],
    \end{split}
    \label{eq:icassp2673}
\eeq
where the left KLD parameter $b_{f_{n}}(\theta)$ is defined in \eqref{eq:ICASSP2636true}. Expanding the KLD and discarding the constants gives
\beq
    \begin{split}
        \left(\nu^{f\to \theta}_n, \xi^{f\to \theta}_n\right)=\arg\!\!\!\!\!\!\!\min_{\left(\nu^{f\to \theta}_n, \xi^{f\to \theta}_n\right)\in \Omega_{n}} L_n\left(\nu^{f\to \theta}_n, \xi^{f\to \theta}_n\right),
    \end{split}
    \label{eq:icassp2674}
\eeq
where
\beq
    \begin{split}
        L_n\left(\nu^{f\to \theta}_n, \xi^{f\to \theta}_n\right)=-\frac{\ln\left[\xi^{\bh_{f}}_n\left(\nu^{f\to \theta}_n, \xi^{f\to \theta}_n\right)\right]}{2}\\
        +\frac{\left[\nu^{\bh_{f}}_n\left(\nu^{f\to \theta}_n, \xi^{f\to \theta}_n\right)-\xi^{\bh_{f}}_n\left(\nu^{f\to \theta}_n, \xi^{f\to \theta}_n\right)\mu^{b_{f}}_n\right]^2}{2\xi^{\bh_{f}}_n\left(\nu^{f\to \theta}_n, \xi^{f\to \theta}_n\right)}\\
        +\frac{\xi^{\bh_{f}}_n\left(\nu^{f\to \theta}_n, \xi^{f\to \theta}_n\right)\tau^{b_{f}}_n}{2}.
    \end{split}
    \label{eq:icassp2675}
\eeq

For the unconstrained case $\Omega_n=\mathbb{R}^2$, the optimal point of \eqref{eq:icassp2674} is obtained by setting the partial derivative with respect to each of the parameters to zero:
\beq
    \begin{split}
        \xi^{f\to \theta}_{n}&=\frac{1}{\tau^{b_{f}}_{n}}-\xi^{f\leftarrow \theta}_{n}\\
        \nu^{f\to \theta}_{n}&=(\xi^{f\to \theta}_{n}+\xi^{f\leftarrow \theta}_{n})\mu^{b_{f}}_{n}-\nu^{f\leftarrow \theta}_n\\
        &=\frac{\mu^{b_{f}}_{n}}{\tau^{b_{f}}_n}-\nu^{f\leftarrow \theta}_n.
    \end{split}
\eeq
Compared to the mean-variance representation, the natural-parameter representation allows smooth transition as $\xi^{f\to \theta}_n$ crosses zero. Thus, $\xi^{f\to \theta}_n$ and $\nu^{f\to \theta}_n$ remain well-defined in $\mathbb{R}$ as long as the belief $b_{f_{n}}(\theta)$ is integrable with well-defined mean and variance.

Finally, we approximate the mean and variance of the given univariate factored PDF using those of the variable-level belief:
\beq
    b_{\theta}(\theta)=\mathcal{M}(\theta|\nu^{b_{\theta}}, \xi^{b_{\theta}})\propto \prod_{n}\Delta^{f\to \theta}_{n}(\theta),
    \label{eq:icassp2680rg}
\eeq
with 
\beq
\begin{split}
    &\nu^{b_{\theta}}=\sum_{n}\nu^{f\to \theta}_{n}\\
    &\xi^{b_{\theta}}=\sum_{n}\xi^{f\to \theta}_{n}.
\end{split}
\label{eq:icassp2678}
\eeq
The corresponding belief variance and mean are
\beq
\begin{split}
    \tau^{b_{\theta}}=\frac{1}{\xi^{b_{\theta}}}\\
    \mu^{b_{\theta}}=\frac{\nu^{b_{\theta}}}{\xi^{b_{\theta}}}.
\end{split}
\label{eq:icassp2679}
\eeq

\begin{proposition}
    Assume that $b_{\theta}(\theta)$ is defined as in \eqref{eq:icassp2680rg}. If the algorithm in \eqref{eq:icassp2665}-\eqref{eq:icassp2679} is initialized with messages such that $b_{\theta}(\theta)$ is integrable, then it remains integrable at every iteration.
\end{proposition}
\begin{proof}
    As in Proposition~\ref{prop:icassp2601}, we proceed by mathematical induction. By assumption, the initial $b_{\theta}(\theta)$ is integrable.

    As the induction hypothesis, assume that after updating the message $\Delta^{f\to \theta}_{k}(\theta)$, the belief $b_{\theta}(\theta)$ remains integrable. 

    Let $n=(k\mod N)+1$. The algorithm enforces that the factor-level belief $b_{f_{n}}(\theta)$ is integrable. Consequently, the objective function \eqref{eq:icassp2675} is well defined. From the first term on the right-hand side of \eqref{eq:icassp2675}, we see that, for $L_{n}$ to be well defined, it is required that 
    \beq
        \xi^{f\to \theta}_{n}+\xi^{f\leftarrow \theta}_n>0,
    \eeq
    which defines the natural domain for $L_{n}$.

    After updating $\xi^{f\to \theta}_{n}$, the precision (inverted variance) of the variable-level belief becomes
    \beq
        \xi^{b_{\theta}}=\xi^{f\to \theta}_{n}+\xi^{f\leftarrow \theta}_n>0.
    \eeq
    Therefore, the belief $b_{\theta}(\theta)$ remains integrable.
\end{proof}

\subsection{GMM Case}

As an illustrative example, consider Gaussian mixture factors of the form given in \eqref{eq:ICASSP2619}
\beq
f_n(\theta)\propto \sum_{s_n} p_{s_n} \mathcal{M}(\theta|\nu_{s_n}, \xi_{s_n}),\label{eq:ICASSP2636}
\eeq
where
\beq
    \begin{split}
    \nu_{s_{n}}=\frac{\mu_{s_{n}}}{\tau_{s_{n}}};\\
    \xi_{s_{n}}=\frac{1}{\tau_{s_{n}}}.
\end{split}
\eeq
From \eqref{eq:ICASSP2636true} and \eqref{eq:ICASSP2632}, the factor-level belief at $f_{n}$ can be written as
\beq
\begin{split}
    b_{f_{n}}(\theta)\propto \sum_{s_n} \omega_{s_n}\mathcal{M}(\theta|\nu_{s_{n}}+\nu^{f\leftarrow \theta}_{n}, \xi_{s_{n}}+\xi^{f\leftarrow \theta}_{n}),
    \label{eq:ICASSP2643t}
\end{split}
\eeq
where
\beq
    \omega_{s_n}\!\!\!=\!\frac{p_{s_n}\sqrt{\xi_{s_n}}}{\sqrt{\xi_{s_{n}}\!\!\!+\!\xi^{f\leftarrow \theta}_{n}}}\mathcal{UM}\!\!\left(\!0\!\left|\frac{\xi^{f\leftarrow \theta}_{n}\nu_{s_n}\!\!\!-\!\xi_{s_n}\nu^{f\leftarrow \theta}_{n}}{\xi_{s_n}+\xi^{f\leftarrow \theta}_{n}}, \frac{\xi_{s_n}\xi^{f\leftarrow \theta}_{n}}{\xi_{s_n}\!\!\!+\!\xi^{f\leftarrow \theta}_{n}}\!\!\right.\right).
    \label{eq:ICASSP2643}
\eeq

As discussed earlier, $\xi^{f\leftarrow \theta}_{n}$ may be zero, which is a discontinuity in \eqref{eq:ICASSP2643}. However, due to the design of the additional domain $\Omega_{n}$ as we will see in \eqref{eq:icassp2690df} and 
\eqref{eq:icassp2690dfs}, we use the ansatz that $\forall s_n: \xi_{s_n}\!\!\!+\!\xi^{f\leftarrow \theta}_{n}\geq 0$. To analyze the asymptotic behavior of \eqref{eq:ICASSP2643t}, we examine the ratio between $\omega_{s_n'}$ and $\omega_{s_n}$
\beq
\begin{split}
    &\rho^{s_{n}'}_{s_n}=\frac{\omega_{s_n'}}{\omega_{s_n}}=\frac{p_{s_{n}'}\sqrt{\xi_{s'_{n}} \left(\xi_{s_{n}}+\xi^{f\leftarrow \theta}_{n}\right)}}{p_{s_{n}}\sqrt{\xi_{s_{n}} \left(\xi_{s'_{n}}+\xi^{f\leftarrow \theta}_{n}\right)}} \\
    &\cdot\exp\!\!\left(\!\!-\frac{1}{2}\!\!\left[\frac{\xi^{f\leftarrow\theta}_{n}\nu_{s'_{n}}^2}{\xi_{s'_{n}}\left(\xi_{s'_{n}}+\xi^{f\leftarrow\theta}_{n}\right)}-\frac{2\nu_{s'_{n}}\nu^{f\leftarrow\theta}_{n}}{\xi_{s'_{n}}+\xi^{f\leftarrow\theta}_{n}}+\frac{2\nu_{s_{n}}\nu^{f\leftarrow\theta}_{n}}{\xi_{s_{n}}+\xi^{f\leftarrow\theta}_{n}}\right.\right.\\
    &\left.\left.-\frac{\xi^{f\leftarrow \theta}_{n}\nu_{s_{n}}^2}{\xi_{s_{n}}\left(\xi_{s_{n}}+\xi^{f\leftarrow \theta}_{n}\right)}+\frac{\nu^{f\leftarrow \theta}_{n}}{\xi_{s_{n}}+\xi^{f\leftarrow \theta}_{n}}-\frac{\nu^{f\leftarrow\theta}}{\xi_{s'_n}+\xi^{f\leftarrow \theta}_{n}}\right]\right).\\
\end{split}
\label{eq:ICASSP2645}
\eeq
From \eqref{eq:ICASSP2645}, we observe that $\rho^{s_{n}'}_{s_n}$ remains finite and positive if $\xi^{f\leftarrow \theta}_{n}=0$. If $\xi^{f\leftarrow \theta}_{n}\to (-\xi_{s'_n})^+$, we have $\rho^{s_{n}'}_{s_n}\to 0$. Moreover, we can verify that if $\xi^{f\leftarrow \theta}_{n}\to (-\xi_{s'_n})^+$, $\rho^{s_{n}'}_{s_n}$ is a higher order infinitesimal of $\xi^{f\leftarrow \theta}_{n} +\xi_{s'_n}$. A similar asymptotic behavior at $\xi^{f\leftarrow \theta}_{n}\to (-\xi_{s_n})^+$ can be obtained by analyzing $\rho^{s_{n}}_{s'_n}$.
Consequently, \eqref{eq:ICASSP2643t} can be equivalently written as 
\beq
    \forall s_{n}: b_{f_{n}}(\theta)\propto \sum_{s'_n} \rho^{s'_n}_{s_n}\mathcal{M}(\theta|\nu_{s'_{n}}+\nu^{f\leftarrow \theta}_{n}, \xi_{s'_{n}}+\xi^{f\leftarrow \theta}_{n}).
    \label{eq:ICASSP2646}
\eeq
To normalize \eqref{eq:ICASSP2646}, we define the normalized ratio 
\beq
    \rhob^{s'_n}_{s_n}=\frac{\rho^{s'_n}_{s_n}}{\sum_{s'_{n}}{\rho^{s'_{n}}_{s_{n}}}}.
    \label{eq:icassp2688}
\eeq
The mean and variance of $b_{f_{n}}(\theta)$ in \eqref{eq:ICASSP2646} are given by
\beq
    \begin{split}
    &\mu^{b_{f}}_{n}=\sum_{s'_{n}}\frac{\rhob^{s'_{n}}_{s_{n}}\left(\nu_{s'_{n}}+\nu^{f\leftarrow \theta}_{n}\right)}{\xi_{s'_{n}}+\xi^{f\leftarrow \theta}_{n}}\\
    &\tau^{b_{f}}_{n}\!=\!\left(\!\sum_{s'_{n}}\rhob^{s'_n}_{s_n}\!\!\left[\!\frac{1}{\xi_{s'_{n}}+\xi^{f\leftarrow \theta}_{n}}\!+\!\!\left(\!\frac{\nu_{s'_{n}}\!+\!\nu^{f\leftarrow \theta}_{n}}{\xi_{s'_{n}}\!+\!\xi^{f\leftarrow \theta}_{n}}\right)^2\right]\right)\!\!-\!(\mu^{b_{f}}_{n})^2.
\end{split}
\label{eq:icassp2689}
\eeq
To close the loop, we update the message $\Delta^{f\to \theta}_{n}(\theta)$ using \eqref{eq:icassp2673} - \eqref{eq:icassp2675}. 
Let the next factor be $k=(n\mod N)+1$. Since the integrability of a GMM is determined by the signs of its component variances/precisions, we define the threshold for $\xi^{f\to \theta}_{n}$ as
\beq
\begin{split}
    \xi^{\mathrm{thres}}_{n}=-\min_{s_k}(\xi_{s_k})-\prod_{\substack{n'\neq n,\,  n'\neq k}}\xi^{f\to \theta}_{n'}\\
    =-\min_{s_k}(\xi_{s_k})-\xi^{f\leftarrow \theta}_{n}+\xi^{f\to \theta}_{k}.
\end{split}
\label{eq:icassp2690df}
\eeq
Thus, the admissible domain of $\left(\nu^{f\to \theta}_n, \xi^{f\to \theta}_n\right)$ that ensures the integrability of $b_{f_{k}}(\theta)$ is
\beq
    \Omega_{n}=\mathbb{R}\times(\xi^{\mathrm{thres}}_{n}, +\infty),
    \label{eq:icassp2690dfs}
\eeq
where $(\xi^{\mathrm{thres}}_{n}, \infty)$ denotes an open interval from $\xi^{\mathrm{thres}}_{n}$ to $+\infty$. However, we can include $\xi^{\mathrm{thres}}_{n}$ by using the asymptotic analysis above. 

To solve the KLD projection in \eqref{eq:icassp2673}, since $\nu^{f\to \theta}_n$ is unconstrained, we first set the partial derivative with respect to $\nu^{f\to \theta}_n$ to zero, which gives us
\beq
    \nu^{f\to \theta}_{n}=(\xi^{f\to \theta}_{n}+\xi^{f\leftarrow \theta}_{n})\mu^{b_{f}}_{n}-\nu^{f\leftarrow \theta}_n.
    \label{eq:icassp2687}
\eeq
Substituting \eqref{eq:icassp2687} into \eqref{eq:icassp2675}, we can verify that $L_n\left(\nu^{f\to \theta}_n, \xi^{f\to \theta}_n\right)$ is decreasing for 
\beq
\xi^{f\to \theta}_n\in \left(-\xi^{f\leftarrow \theta}_n, \frac{1}{\tau^{b_{f}}_{n}}-\xi^{f\leftarrow \theta}_{n}\right],
\eeq
and increasing for
\beq
    \xi^{f\to \theta}_n\in \left(\frac{1}{\tau^{b_{f}}_{n}}-\xi^{f\leftarrow \theta}_{n}, +\infty\right).
\eeq
Therefore, $\xi^{f\to \theta}_n$ is obtained as
\beq
    \xi^{f\to \theta}_n=\begin{cases}
        \frac{1}{\tau^{b_{f}}_{n}}-\xi^{f\leftarrow \theta}_{n}, \; &\text{if $\frac{1}{\tau^{b_{f}}_{n}}-\xi^{f\leftarrow \theta}_{n}> \xi^{\mathrm{thres}}_{n}$}\\
        \xi^{\mathrm{thres}}_{n}, \; &\text{otherwise}
    \end{cases}.
    \label{eq:icassp2690}
\eeq
By substituting \eqref{eq:icassp2690} back into \eqref{eq:icassp2687}, we obtain the update for the complete message $\Delta^{f\to \theta}_{n}(\theta)$.

After the final iteration, the mean and variance of the product of GMM factors can be obtained from the variable-level belief
\beq
    b_{\theta}(\theta)=\prod_{n}\Delta^{f\to \theta}_{n}(\theta).
    \label{eq:icassp2691}
\eeq
Accordingly, the mean and variance of $b_{\theta}(\theta)$ can be explicitly computed as
\beq
\begin{split}
    \mu^{b_{\theta}}=\frac{\sum_{n} \nu^{f\to \theta}_{n}}{\sum_{n} \xi^{f\to \theta}_{n}};\\
    \tau^{b_{\theta}}=\frac{1}{\sum_{n} \xi^{f\to \theta}_{n}}.
\end{split}
\label{eq:icassp2697}
\eeq

This concludes the ACEP algorithm for Gaussian mixture model factors, as summarized in Algorithm~\ref{algo:icassp2603}.

\floatstyle{spaceruled}
\restylefloat{algorithm}
\begin{algorithm}[t]
\caption{ACEP for GMM Factors}\label{algo:icassp2603}
\begin{algorithmic}[1]
\State Initialize $\forall n: \nu^{f\to \theta}_{n}=0, \xi^{f\to \theta}_{n}=1$
\Repeat
\State [For every $n\in [1,\dots, N]$, repeat the following]
\State Update $\nu^{b_{\theta}}$ and $\xi^{b_{\theta}}$ by \eqref{eq:icassp2678}
\State $\xi^{f\leftarrow \theta}_{n}=\xi^{b_{\theta}}-\xi^{f\to \theta}_n$
\State $\nu^{f\leftarrow \theta}_{n}=\nu^{b_{\theta}}-\nu^{f\to \theta}_n$
\State Choose the index of a GMM component $s_n$.
\State $\forall s'_{n}:$ compute $\rho^{s'_{n}}_{s_{n}}$ by \eqref{eq:ICASSP2645}
\State $\forall s'_{n}:$ compute $\rhob^{s'_{n}}_{s_{n}}$ by \eqref{eq:icassp2688}
\State Update $\mu^{b_f}_{n}$ and $\tau^{b_{f}}_{n}$ by \eqref{eq:icassp2689}
\State $k=(n\mod N)+1$
\State Update $\xi^{\mathrm{thres}}_{n}$ by \eqref{eq:icassp2690df}
\State Update $\xi^{f\to \theta}_{n}$ by \eqref{eq:icassp2690}
\State Update $\nu^{f\to \theta}_{n}$ by \eqref{eq:icassp2687}

\Until{Convergence}
\State Output $\mu^{b_{\theta}}$ and $\tau^{b_{\theta}}$ as mean and variance by \eqref{eq:icassp2697}
\end{algorithmic}
\end{algorithm}

\section{Simulation Results \label{sect:SIMULATION}}
To verify the effectiveness of the proposed methods, we generate factored PDFs consisting of $N=8$ GMM factors in each realization. Each GMM factor contains two Gaussian components with randomly assigned weights, means, and variances. The product of these $8$ GMM factors results in another GMM with $2^{8}$ components. To assess the approximation accuracy, we compute the exact mean and variance of the resulting $2^{8}$-component GMM using brute force, denoted by $\mu^{\thetah}_{GMM}$ and $\tau^{\thetah}_{GMM}$ respectively.

The normalized squared error (NSE) for a single realization is defined as
\beq
\begin{split}
    &NSE_{\mu}=\frac{(\mu^{\thetah}-\mu^{\thetah}_{GMM})^2}{(\mu^{\thetah}_{GMM})^2}\\
    &NSE_{\tau}=\frac{(\tau^{\thetah}-\tau^{\thetah}_{GMM})^2}{(\tau^{\thetah}_{GMM})^2}.
\end{split}
\label{eq:ICASSP2630}
\eeq
We simulated $10{,}000$ realizations with different sets of GMM factors.
The proposed methods are compared with the commonly used Clipping EP \cite{8713501, 10378663}. Since checking the integrability of the belief $b_{f_{n}}(\theta)$ can be computationally expensive, we consider two variants for both proposed EP-based algorithms:
\begin{itemize}
    \item \textbf{Strict approach (STR.)}: the algorithms are implemented exactly as proposed.
    \item \textbf{Relaxed approach (RLX.)}: for persistent EP, we check the integrability of $\Delta^{f\leftarrow \theta}_{n}(\theta)$ instead of $b_{f_{n}}(\theta)$. For ACEP, we simply set $\xi^{\mathrm{thres}}_{n}$ to zero.
\end{itemize}
The simulation results for mean estimations are shown in Fig.~\ref{fig:results_mean}, while the results for variance estimation are shown in Fig.~\ref{fig:results_var}. The abbreviation ``PEP'' denotes persistent EP.
\begin{figure}[t]
    \centering
    \includegraphics[width=0.48\textwidth]{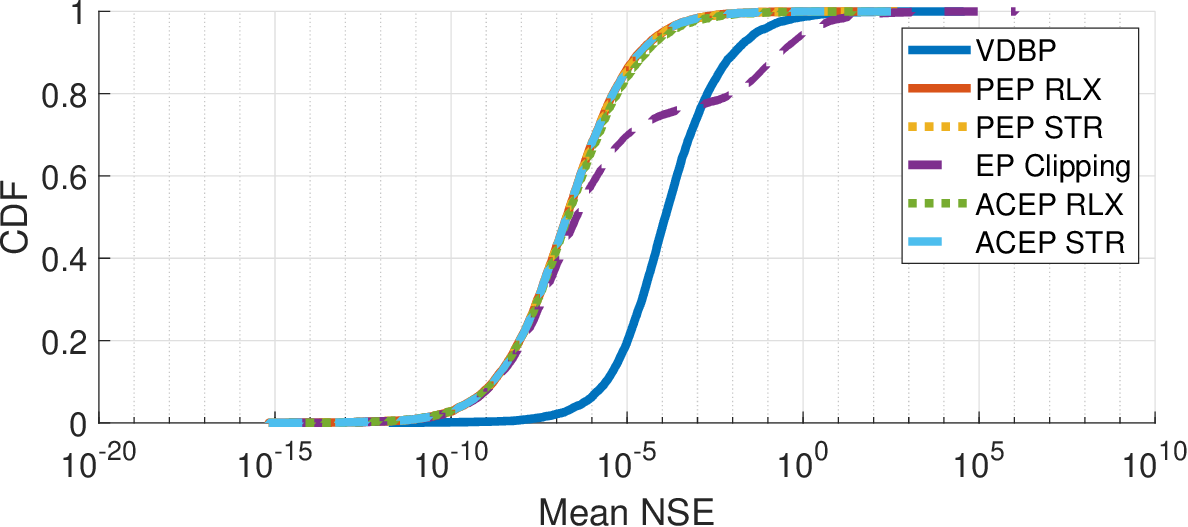}
    \vspace{-2mm}
    \caption{NSE CDF of mean estimates.} 
    \vspace{-2mm}
    \label{fig:results_mean}
\end{figure}

\begin{figure}[t]
    \centering
    \includegraphics[width=0.48\textwidth]{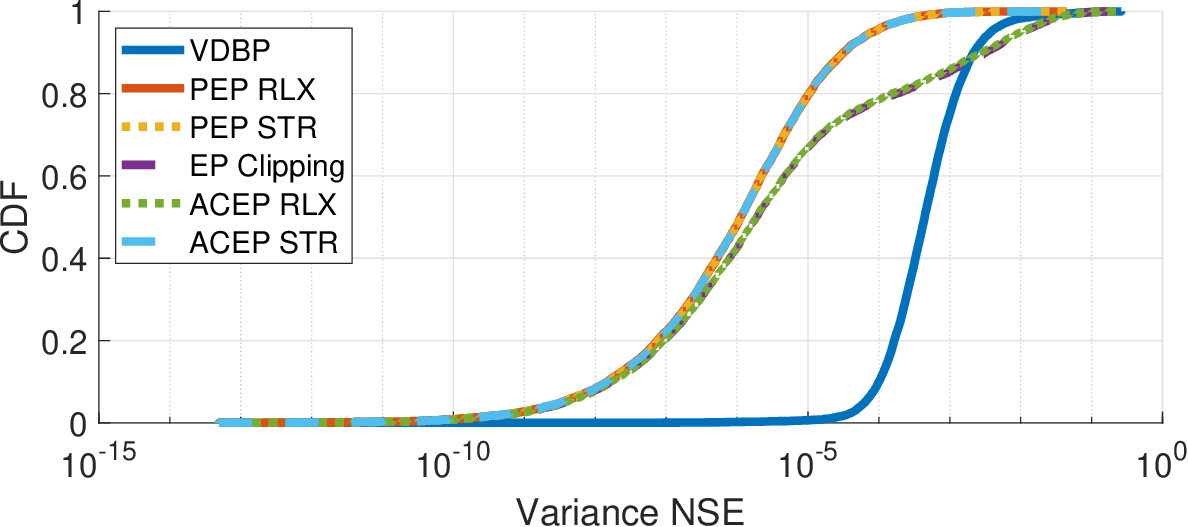}
    \vspace{-2mm}
    \caption{NSE CDF of variance estimates.} 
    \vspace{-2mm}
    \label{fig:results_var}
    \vspace{-4mm}
\end{figure}

In both figures, the curves corresponding to strict persistent EP, relaxed persistent EP, and strict ACEP coincide and consistently exhibit the best performance. All proposed methods outperform clipping EP in terms of NSE at the $95\%$ CDF.

An interesting observation is that the mean estimate produced by relaxed ACEP is close to the best-performing methods, whereas its variance estimate is close to that of clipping EP. This behavior can be explained by comparing relaxed ACEP with strict ACEP and clipping EP. The variance estimation is closely related to the natural parameter $\xi$. Both relaxed ACEP and clipping EP disallow non-integrable messages. Therefore, Fig.~\ref{fig:results_var} suggests that the optimal set of messages may include non-integrable messages. By examining the updates of the mean and natural parameter $\nu$ in clipping EP and ACEP, we observe that ACEP exploits the limit around the discontinuity $\xi^{f\to \theta}_{n}=0$. In contrast, in clipping EP, when $\tau^{f\to \theta}_{n}$ is clipped to $+\infty$, $\mu^{f\to \theta}_{n}$ is still treated as a finite value. However, as indicated by ACEP, $\mu^{f\to \theta}_{n}$ should instead scale as $O(\tau^{f\to \theta}_{n})$.

\section{Conclusions}
This paper proposes one BP-based algorithm and two EP-based algorithms for estimating the means and variances of univariate factored PDFs. 

In the BP-based method, VDBP, we construct a multivariate measurement model whose posterior can be marginalized to the given univariate factored PDF. GaBP is then applied to this multivariate measurement model to approximate the marginal posterior. 

EP can be directly applied to estimate the mean and variance of a univariate factored PDF when all factor-level beliefs remain integrable throughout the iterations. To address scenarios involving non-integrable factor-level beliefs, we proposed persistent EP and ACEP.

Persistent EP adopts a passive strategy by skipping message updates around a factor whenever the corresponding belief at the factor is non-integrable. In contrast, ACEP takes an active approach by modifying the projection family to ensure that the belief at the next factor is always integrable.

{\bf Acknowledgements}
EURECOM's research is partially supported by its industrial members:
ORANGE, BMW, SAP, iABG,  Norton LifeLock, by the Franco-German projects CellFree6G and 5G-OPERA, the French PEPR-5G projects PERSEUS and YACARI, the EU H2030 project CONVERGE, and by a Huawei France funded Chair towards Future Wireless Networks.



{\renewcommand{\baselinestretch}{1.3}\Large
\bibliographystyle{IEEEtran}
\bibliography{ICASSPGMM}

@misc{minka2013expectationpropagationapproximatebayesian,
      title={Expectation Propagation for approximate Bayesian inference}, 
      author={Thomas P. Minka},
      year={2013},
      eprint={1301.2294},
      archivePrefix={arXiv},
      primaryClass={cs.AI},
      url={https://arxiv.org/abs/1301.2294}, 
}

@book{pearl2014probabilistic,
  title={Probabilistic reasoning in intelligent systems: networks of plausible inference},
  author={Pearl, Judea},
  year={2014},
  publisher={Elsevier}
}

@BOOK{8187302,
  author={Wainwright, Martin J. and Jordan, Michael I.},
  title={Graphical Models, Exponential Families, and Variational Inference},
  year={2008},
  doi={10.1561/2200000001}}

@misc{murphy2013loopybeliefpropagationapproximate,
      title={Loopy Belief Propagation for Approximate Inference: An Empirical Study}, 
      author={Kevin Murphy and Yair Weiss and Michael I. Jordan},
      year={2013},
      eprint={1301.6725},
      archivePrefix={arXiv},
      primaryClass={cs.AI},
      url={https://arxiv.org/abs/1301.6725}, 
}

@article{heskes2005approximate,
  title={Approximate inference techniques with expectation constraints},
  author={Heskes, Tom and Opper, Manfred and Wiegerinck, Wim and Winther, Ole and Zoeter, Onno},
  journal={Journal of Statistical Mechanics: Theory and Experiment},
  volume={2005},
  number={11},
  pages={P11015},
  year={2005},
  publisher={IOP Publishing}
}

@article{roche2025affine,
  title={Affine Filter Bank Modulation (AFBM): A Novel 6G ISAC Waveform with Low PAPR and OOBE},
  author={Roche Rayan Ranasinghe, Kuranage and Senger, Henrique L and Gon{\c{c}}alves, Gustavo P and Rou, Hyeon Seok and Chang, Bruno S and Thadeu Freitas de Abreu, Giuseppe and Le Ruyet, Didier},
  journal={arXiv e-prints},
  pages={arXiv--2509},
  year={2025}
}

@ARTICLE{9460784,
  author={Iimori, Hiroki and Takahashi, Takumi and Ishibashi, Koji and de Abreu, Giuseppe Thadeu Freitas and Yu, Wei},
  journal={IEEE Transactions on Wireless Communications}, 
  title={Grant-Free Access via Bilinear Inference for Cell-Free MIMO With Low-Coherence Pilots}, 
  year={2021},
  volume={20},
  number={11},
  pages={7694-7710},
  doi={10.1109/TWC.2021.3088125}}

@ARTICLE{8496782,
  author={Zou, Qiuyun and Zhang, Haochuan and Wen, Chao-Kai and Jin, Shi and Yu, Rong},
  journal={IEEE Signal Processing Letters}, 
  title={Concise Derivation for Generalized Approximate Message Passing Using Expectation Propagation}, 
  year={2018},
  volume={25},
  number={12},
  pages={1835-1839},
  doi={10.1109/LSP.2018.2876806}}

@ARTICLE{10378663,
  author={Karataev, Alexander and Forsch, Christian and Cottatellucci, Laura},
  journal={IEEE Open Journal of Signal Processing}, 
  title={Bilinear Expectation Propagation for Distributed Semi-Blind Joint Channel Estimation and Data Detection in Cell-Free Massive MIMO}, 
  year={2024},
  volume={5},
  number={},
  pages={284-293},
  doi={10.1109/OJSP.2023.3348343}}

@misc{rangan2018vectorapproximatemessagepassing,
      title={Vector Approximate Message Passing}, 
      author={Sundeep Rangan and Philip Schniter and Alyson K. Fletcher},
      year={2018},
      eprint={1610.03082},
      archivePrefix={arXiv},
      primaryClass={cs.IT},
      url={https://arxiv.org/abs/1610.03082}, 
}

@ARTICLE{6898015,
  author={Parker, Jason T. and Schniter, Philip and Cevher, Volkan},
  journal={IEEE Transactions on Signal Processing}, 
  title={Bilinear Generalized Approximate Message Passing—Part I: Derivation}, 
  year={2014},
  volume={62},
  number={22},
  pages={5839-5853},
  doi={10.1109/TSP.2014.2357776}}

@INPROCEEDINGS{10942970,
  author={Zhao, Zilu and Slock, Dirk},
  booktitle={2024 58th Asilomar Conference on Signals, Systems, and Computers}, 
  title={Decentralized Message-Passing for Semi-Blind Channel Estimation in Cell-Free Systems Based on Bethe Free Energy Optimization}, 
  year={2024},
  volume={},
  number={},
  pages={1443-1447},
  doi={10.1109/IEEECONF60004.2024.10942970}}

@inproceedings{schrempf2005optimal,
  title={Optimal mixture approximation of the product of mixtures},
  author={Schrempf, Oliver C and Feiermann, Olga and Hanebeck, Uwe D},
  booktitle={2005 7th International Conference on Information Fusion},
  volume={1},
  pages={8--pp},
  year={2005},
  organization={IEEE}
}

@article{ihler2003efficient,
  title={Efficient multiscale sampling from products of Gaussian mixtures},
  author={Ihler, Alexander and Sudderth, Erik and Freeman, William and Willsky, Alan},
  journal={Advances in Neural Information Processing Systems},
  volume={16},
  year={2003}
}

@ARTICLE{8713501,
  author={Rangan, Sundeep and Schniter, Philip and Fletcher, Alyson K.},
  journal={IEEE Transactions on Information Theory}, 
  title={Vector Approximate Message Passing}, 
  year={2019},
  volume={65},
  number={10},
  pages={6664-6684},
  doi={10.1109/TIT.2019.2916359}}

@INPROCEEDINGS{1211409,
  author={Sudderth, E.B. and Ihler, A.T. and Freeman, W.T. and Willsky, A.S.},
  booktitle={2003 IEEE Computer Society Conference on Computer Vision and Pattern Recognition, 2003. Proceedings.}, 
  title={Nonparametric belief propagation}, 
  year={2003},
  volume={1},
  number={},
  pages={I-I},
  doi={10.1109/CVPR.2003.1211409}}

@misc{zhao2025expectationsexpectationpropagation,
      title={Expectations in Expectation Propagation}, 
      author={Zilu Zhao and Fangqing Xiao and Dirk Slock},
      year={2025},
      eprint={2512.08034},
      archivePrefix={arXiv},
      primaryClass={cs.IT},
      url={https://arxiv.org/abs/2512.08034}, 
}
}


\end{document}